\documentclass[a4paper, 12pt, conference]{ieeeconf}      % Use this line for a4 paper

\IEEEoverridecommandlockouts                              % This command is only needed if
\title{\Large\bf Robust Mean Square Stability of Open Quantum Stochastic Systems with Hamiltonian Perturbations in a Weyl Quantization Form}
\author{Arash Kh. Sichani, \qquad Igor G. Vladimirov, \qquad Ian R. Petersen%\\ \today, \currenttime\\
\thanks{This work is supported by the Australian Research Council. The authors are with UNSW Canberra, ACT 2600, Australia. E-mail: {\tt arash\_kho@hotmail.com, igor.g.vladimirov@gmail.com, i.r.petersen@gmail.com}.}
}

\usepackage{mathptmx} % assumes new font selection scheme installed

\usepackage{graphicx} % for pdf, bitmapped graphics files
\usepackage{times} % assumes new font selection scheme installed
\usepackage{amsmath} % assumes amsmath package installed
\usepackage{amssymb}  % assumes amsmath package installed
\usepackage{datetime}
\usepackage{comment}
\newtheorem{lem}{Lemma}
\newtheorem{thm}{Theorem}

\newtheorem{exmp}{Example}
\def\<{\leqslant}           % nice less than or equal to sign
\def\>{\geqslant}           % nice larger than or equal to sign
\def\d{\partial}

\def\wt{\widetilde}
\def\Re{\mathrm{Re}}   % real part
\def\mA{\mathbb{A}}    % space of real antisymmetric matrices
\def\mR{\mathbb{R}}    % real line
\def\mC{\mathbb{C}}    % complex plane

\def\Tr{\mathrm{Tr}}       % matrix trace
\def\rT{\mathrm{T}}        % matrix transpose
\def\bE{\mathbf{E}}    % expectation

\def\[[[{[\![\![}
\def\]]]{]\!]\!]}

\def\bra{\langle }
\def\ket{\rangle }

\def\Bra{\left\langle }
\def\Ket{\right\rangle }

\def\re{\mathrm{e}}        % number e
\def\rd{\mathrm{d}}        % differential

\def\x{\times}
\def\ox{\otimes}

\def\cP{\mathcal{P}}
\def\cU{\mathcal{U}}

\def\mH{{\mathbb H}}
\def\mS{{\mathbb S}}

\onecolumn
\begin{document}
\maketitle
\thispagestyle{empty}
\pagestyle{plain}

%%%%%%%%%%%%%%%%%%%%%%%%%%%%%%%%%%%%%%%%%%%%%%%%%%%%%%%%%%%%%%%%%%%%%%%%%%%%%%%%%%%%%%%%%%%%%%%%%%%%
\begin{abstract}
This paper is concerned with open quantum systems whose dynamic variables satisfy canonical commutation relations and are governed by quantum stochastic differential equations. The latter are driven  by quantum Wiener processes which represent external boson fields. The system-field coupling operators are linear functions of the system variables. The Hamiltonian consists of a nominal quadratic function of the system variables and an uncertain perturbation which is represented in a Weyl quantization form. Assuming that the nominal linear quantum system is stable, we develop sufficient conditions on the perturbation of the Hamiltonian which guarantee robust mean square stability of the perturbed system. Examples are given to illustrate these results for a class of Hamiltonian perturbations in the form of trigonometric polynomials of the system variables.
\end{abstract}
\begin{keywords}%
open quantum stochastic system,
Hamiltonian perturbation,
Weyl quantization,
robust mean square stability.
\end{keywords}

%%%%%%%%%%%%%%%%%%%%%%%%%%%%%%%%%%%%%%%%%%%%%%%%%%%%%%%%%%%%%%%%%%%%%%%%%%%%%%%%%%%%%%%%%%%%%%%%%%%
\section{INTRODUCTION}
%%%%%%%%%%%%%%%%%%%%%%%%%%%%%%%%%%%%%%%%%%%%%%%%%%%%%%%%%%%%%%%%%%%%%%%%%%%%%%%%%%%%%%%%%%%%%%%%%%%

The quantum mechanical concept of quantization is concerned with assigning quantum observables to classical variables (and functions thereof). Weyl's proposal for the development of a general quantization  scheme was introduced  in 1927   (see for example, \cite[Section IV.14]{weyl1950theory})  soon after the invention of quantum mechanics. An important feature of the Weyl association is that it treats quantum dynamic variables equally and leads to correct marginal distributions for them \cite[Chapter 8]{dubin2000}. The Weyl quantization scheme employs Fourier transforms and is known to be a convenient and, in many respects, satisfactory procedure for  quantization \cite{dubin2000,F_1989,zachos2005}. In addition to providing a mathematical formalism, this scheme also offers an interpretation of quantum mechanical phenomena, thus leading to a better understanding of their physical aspects \cite{dubin2000,zachos2005}. The aim of the present paper is to use the Weyl quantization for the modeling of perturbations of  Hamiltonians for a class of open quantum systems and the robust stability analysis based on this description of uncertainty.

A wide range of open quantum systems, which interact with their environment, can be modelled by using the apparatus of quantum stochastic differential equations (QSDEs) \cite{HP_1984,PKR_2012}. In this framework, which follows the Heisenberg picture of quantum dynamics, a quantum noise is introduced in order to represent the surroundings as a heat bath of external fields acting on a boson Fock space \cite{PKR_2012}. The QSDE approach to open quantum systems is employed by the quantum dissipative systems theory \cite{J2010diss}  which addresses robust stability issues.

%The theory of robust or absolute stability in the modern control theory provides a framework to construct a bridge between control theory and its applications. The capability of extending the results of simple analysis and design approaches to real systems with uncertainties is the main subject of robust control theory. In the field of stochastic quantum control theory, perturbations, as an important class of uncertainties, has been studied in \cite{}. The link between the admissible ...

The robustness of various classes of perturbed open quantum systems, modelled by QSDEs, has been addressed in the literature using dissipativity theory and different notions of stability (see for example \cite{IVM2012,Ian2012robust,VP_2012a,VP_2012b}). In particular, robust mean square stability with respect to  a class of perturbations of  Hamiltonians has been studied in \cite{IVM2012} and its applications have been presented in \cite{Ian_ker_2013,Ian_Josephson_2012}. In these papers, the classical and quantum models of the perturbed Hamiltonian are related by using power series of quantum variables in combination with Wick's quantization \cite[pp. 445]{dubin2000}.

In the present paper, we consider a class of open quantum systems whose dynamic variables satisfy Heisenberg canonical commutation relations (CCRs) and are governed by QSDEs. The system-field coupling operators are assumed to be known linear functions of the system variables, while the Hamiltonian is split into a nominal quadratic part and an uncertain perturbation which is, in general, a non-quadratic function of the system variables. Following a similar approach in \cite{V_2014}, we use Weyl quantization in order to model the perturbation of  the system Hamiltonian.  The fact that the Weyl quantization employs Fourier transforms makes it particularly suitable for modelling uncertainties in the form of  trigonometric polynomials of the system variables, such as in \cite{Ian_Josephson_2012}. Assuming that the nominal linear quantum stochastic  system is stable, we develop sufficient conditions for the robust mean square stability of the perturbed system. These conditions employ a linear matrix inequality (LMI) which, in addition to the conventional Lyapunov part, involves a linear operator acting on matrices, whose structure is analogous to that of generalized  Sylvester equations (see, for example, \cite{GLAM_1992}). Such operators play an important role for moment stability of quasilinear quantum stochastic systems \cite{VP_2012b}.

The rest of the paper is organized as follows. Section~\ref{sec:system} describes the class of open quantum systems being considered. Section~\ref{sec:non_pertrb_H} models Hamiltonian perturbations in a Weyl quantization form. Sections~\ref{sec:quadro} and \ref{sec:diss} discuss the time evolution of weighted mean square functionals of the system variables and a related dissipation inequality.
Section~\ref{sec:stab} provides sufficient conditions of robust mean square stability for an admissible set of Hamiltonian perturbations. Section~\ref{sec:exinequ} discusses techniques for verifying these conditions in terms of the Weyl quantization model. Section~\ref{sec:exmpl} provides examples to demonstrate applicability of the approach. Section~\ref{sec:conclusion} summarizes the results of the paper.

%\begin{comment}
%Each connection between classical and quantum mechanics corresponds to a choice of spectral theory for noncommuting operators.{chapter 8}[x1] p.431

%Weyl's ideas were later elaborated in papers in the physics litreture such as Groenewold [61], Moyal [108] and Pool [118].[x3]

% x[3] Gerald B.-Folland Harmonic analysis in phase space

%[x2] Mathematical topics between classical and quantum mechanics

%[x1] Dubin- Mathematical aspects of Weyl quantization

%But let us emphasize that as Margenau and Cohen have shown no rule of ordering can be consistently used to derive quantum operators from their classical functions [15].
%\end{comment}

%%%%%%%%%%%%%%%%%%%%%%%%%%%%%%%%%%%%%%%%%%%%%%%%%%%%%%%%%%%%%%%%%%%%%%%%%%%%%%%%
\section{NOTATION}\label{sec:not}
%%%%%%%%%%%%%%%%%%%%%%%%%%%%%%%%%%%%%%%%%%%%%%%%%%%%%%%%%%%%%%%%%%%%%%%%%%%%%%%%

Unless specified otherwise,  vectors are organized as columns, and the transpose $(\cdot)^{\rT}$ acts on matrices with operator-valued entries as if the latter were scalars. For a vector $X$ of operators $X_1, \ldots, X_r$ and a vector $Y$ of operators $Y_1, \ldots, Y_s$, the commutator matrix is defined as an $(r\x s)$-matrix
$
    [X,Y^{\rT}]
    :=
    XY^{\rT} - (YX^{\rT})^{\rT}
$
whose $(j,k)$th entry is the commutator
$
    [X_j,Y_k]
    :=
    X_jY_k - Y_kX_j
$ of the operators $X_j$ and $Y_k$.
Also, $(\cdot)^{\dagger}:= ((\cdot)^{\#})^{\rT}$ denotes the transpose of the entry-wise operator adjoint $(\cdot)^{\#}$. In application to complex matrices,  $(\cdot)^{\dagger}$ reduces to the complex conjugate transpose  $(\cdot)^*:= (\overline{(\cdot)})^{\rT}$. Furthermore, $\mS_r$, $\mA_r$
 and
$
    \mH_r
    :=
    \mS_r + i \mA_r
$ denote
the subspaces of real symmetric, real antisymmetric and complex Hermitian  matrices of order $r$, respectively, with $i:= \sqrt{-1}$ the imaginary unit. Also, $I_r$ denotes the identity matrix of order $r$, positive (semi-) definiteness of matrices is denoted by ($\succcurlyeq$) $\succ$, and $\ox$ is the tensor product of spaces or operators (for example, the Kronecker product of matrices).
The adjoints and  self-adjointness of linear operators acting on matrices is understood in the sense of the Frobenius inner product
$
    \bra M,N\ket
    :=
    \Tr(M^*N)
$ of real or complex matrices, with the corresponding Frobenius norm $\|M\|:= \sqrt{\bra M,M\ket}$ which reduces to the standard Euclidean norm $|\cdot|$ for vectors. Also, $\|v\|_K:= \sqrt{v^{\rT}Kv}$ denotes the Euclidean norm of a real vector $v$ associated with a real positive definite symmetric matrix $K$.
Finally, $\bE \xi := \Tr(\rho \xi)$ denotes the quantum expectation of a quantum variable $\xi$ (or a matrix of such variables) over a density operator $\rho$ which specifies the underlying quantum state. For matrices of quantum variables, the expectation is evaluated entry-wise.

%%%%%%%%%%%%%%%%%%%%%%%%%%%%%%%%%%%%%%%%%%%%%%%%%%%%%%%%%%%%%%%%%%%%%%%%%%%%%%%%%%%%%%%%%%%%%%%%%%%
\section{OPEN QUANTUM STOCHASTIC SYSTEMS}\label{sec:system}
%%%%%%%%%%%%%%%%%%%%%%%%%%%%%%%%%%%%%%%%%%%%%%%%%%%%%%%%%%%%%%%%%%%%%%%%%%%%%%%%%%%%%%%%%%%%%%%%%%%

We consider an open quantum stochastic system interacting with an external boson field.  The system has $n$ dynamic variables $X_1, \ldots, X_n$ which satisfy CCRs
\begin{equation}
\label{xCCR}
    [X, X^{\rT}] = 2i \Theta,
    \qquad
    X:=
    \begin{bmatrix}
        X_1\\
        \vdots\\
        X_n
    \end{bmatrix},
\end{equation}
where the CCR matrix $\Theta \in \mA_n$ is assumed to be non-singular. The system variables evolve in time according to a QSDE
\begin{equation}
\label{dx}
  \rd X = \Big(i[H,X] -\frac{1}{2} BJB^{\rT} \Theta^{-1} X\Big)\rd t+ B \rd W.
\end{equation}
Here, $W$ is an $m$-dimensional vector of quantum Wiener processes $W_1, \ldots, W_m$ with a  positive semi-definite It\^{o} matrix $\Omega \in \mH_m$:
\begin{equation}
\label{WW}
    \rd W \rd W^{\rT}
    =
    \Omega \rd t,
    \qquad
    \Omega = I_m + iJ,
\end{equation}
where $J \in \mA_m$. The matrix $B \in \mR^{n\x m}$ in (\ref{dx}) is related to a matrix $M \in \mR^{m\x n}$ of linear dependence of the system-field coupling operators on the system variables by
$
    B = 2\Theta M^{\rT}
$.
The term $-\frac{1}{2} BJB^{\rT} \Theta^{-1} X$ in the drift of the QSDE (\ref{dx}) is the Gorini-Kossakowski-Sudarshan-Lindblad (GKSL) decoherence superoperator \cite{GKS_1976,L_1976} which acts on the system variables and is associated with the system-field interaction. Also, $H$ is the Hamiltonian which describes the self-energy of the system and is usually represented as a function of the system variables. For what follows, we assume that $H$ is split into two parts:
\begin{equation}
\label{H}
    H = H_0 + H_1.
\end{equation}
Here,
\begin{equation}
\label{H0}
    H_0 := \frac{1}{2} X^{\rT} R X = \frac{1}{2}\sum_{j,k=1}^{n}r_{jk}X_jX_k
\end{equation}
is a quadratic function of the system variables with a real symmetric matrix $R:= (r_{jk})_{1\< j,k\< n}$ of order $n$, which corresponds to a nominal open quantum harmonic oscillator \cite{EB_2005,GZ_2004}. Also, $H_1$ is a self-adjoint operator on the underlying Hilbert space, which is interpreted as a perturbation of the Hamiltonian and is described in Section~\ref{sec:non_pertrb_H}. By substituting  (\ref{H}) and (\ref{H0}) into (\ref{dx}) and using the CCRs (\ref{xCCR}), it follows that the QSDE takes the form
\begin{equation}
\label{dx1}
  \rd X = (A X + Z)\rd t+ B \rd W,
  \qquad
  Z:= i[H_1, X],
\end{equation}
where $Z$ is an $n$-dimensional vector of self-adjoint operators, and $A \in \mR^{n\x n}$ is given by
\begin{equation}\label{A}
    A
    :=
    2\Theta R - \frac{1}{2} BJB^{\rT} \Theta^{-1} .
\end{equation}
It is assumed that the matrix $A$ is Hurwitz, and hence, the nominal open quantum harmonic oscillator is stable. In particular, in the absence of perturbations (that is, when $H_1=0$, and the system is governed by a linear QSDE $\rd X = AX\rd t + B\rd W$),  the system variables have finite steady-state moments of arbitrary order. Also, note that the CCRs (\ref{xCCR}) are  preserved in time for any perturbation of the system Hamiltonian. This is a consequence of the joint unitary evolution of the dynamic variables of the system and its environment, with the CCR preservation being part of the quantum physical realizability conditions \cite{JNP_2008}.

%%%%%%%%%%%%%%%%%%%%%%%%%%%%%%%%%%%%%%%%%%%%%%%%%%%%%%%%%%%%%%%%%%%%%%%%%%%%%%%%%%%%%%%%%%%%%%%%%%%%%%%%%
\section{PERTURBATION OF HAMILTONIAN}\label{sec:non_pertrb_H}
%%%%%%%%%%%%%%%%%%%%%%%%%%%%%%%%%%%%%%%%%%%%%%%%%%%%%%%%%%%%%%%%%%%%%%%%%%%%%%%%%%%%%%%%%%%%%%%%%%%%%%%%%

We will model the perturbation $H_1$ of the system Hamiltonian in (\ref{H}) by using  the Weyl quantization \cite{F_1989}  as follows:
\begin{equation}
\label{H1}
  H_1 := \int_{\mR^n} h(\lambda) \re^{i\lambda^{\rT}X}\rd \lambda.
\end{equation}
Here, $h: \mR^n \to \mC$ is a complex-valued function which satisfies $h(-\lambda) = \overline{h(\lambda)}$ for all $\lambda \in \mR^n$, thus ensuring that $H_1$ is a self-adjoint operator on the underlying Hilbert space. The function $h$ is computed as the standard Fourier transform
\begin{equation}
\label{h}
  h(\lambda) := (2\pi)^{-n}\int_{\mR^n} H_1(x)\re^{-i\lambda^{\rT} x}\rd x
\end{equation}
of a real-valued function $H_1: \mR^n \to \mR$ of $n$ classical variables whose quantization leads to (\ref{H1}).
The following lemma %(see also \cite{V_2014})
computes the perturbation term in the governing QSDE in the presence of perturbations.

%%%%%%%%%%%%%%%%%%%%%%%%%%%%%%%%%%%%%%%%%%%%%%%%%%%%%%%%%%%%%%%%%%%%%%%%%%%%%%%%%%%%%%%%%%%%%%%%%%%
\begin{lem}
\label{lem:Z}
The perturbation vector $Z$ in the drift of the QSDE (\ref{dx1}), which corresponds to (\ref{H1}), can be computed as
\begin{equation}
\label{iH1X}
    Z =
    2i
    \Theta
    \int_{\mR^n}
        h(\lambda)
        \lambda
            \re^{i\lambda^{\rT}X}
    \rd \lambda.
\end{equation}
\end{lem}
%%%%%%%%%%%%%%%%%%%%%%%%%%%%%%%%%%%%%%%%%%%%%%%%%%%%%%%%%%%%%%%%%%%%%%%%%%%%%%%%%%%%%%%%%%%%%%%%%%%

\begin{proof}
By substituting (\ref{H1}) into the definition of $Z$ in (\ref{dx1}) and using the bilinearity of the commutator, it follows that
\begin{equation}
\label{iH1X0}
    Z = i\int_{\mR^n} h(\lambda) [\re^{i\lambda^{\rT}X}, X]\rd \lambda.
\end{equation}
Now,
\begin{align}
\nonumber
    [\re^{i\lambda^{\rT}X}, X]
    & = \re^{i\lambda^{\rT}X} X - X \re^{i\lambda^{\rT}X}\\
\label{expX}
    & =
    \big(\re^{i\lambda^{\rT}X} X \re^{-i\lambda^{\rT}X}  - X\big) \re^{i\lambda^{\rT}X}
    =2\Theta \lambda \re^{i\lambda^{\rT}X},
\end{align}
where use is made of the relation
\begin{align*}
    \re^{i\lambda^{\rT}X} X \re^{-i\lambda^{\rT}X}
    & =
    \re^{[i\lambda^{\rT}X, \cdot]}(X)
    =
    \sum_{k = 0}^{+\infty}
    \frac{1}{k!}
    [i\lambda^{\rT}X, \cdot]^k(X)\\
    & =
    X + [i\lambda^{\rT}X, X]
    =
    X - i[X,X^{\rT}]\lambda
     = X + 2\Theta \lambda
\end{align*}
which follows from Hadamard's lemma \cite{M_1998} and the CCRs (\ref{xCCR}). Here, $[\xi, \cdot]^k$ denotes the $k$-fold application of the commutator with a given operator $\xi$. Substitution of (\ref{expX}) into (\ref{iH1X0}) leads to (\ref{iH1X}).
\end{proof}
%%%%%%%%%%%%%%%%%%%%%%%%%%%%%%%%%%%%%%%%%%%%%%%%%%%%%%%%%%%%%%%%%%%%%%%%%%%%%%%%%%%%%%%%%%%%%%%%%%%

Note that the vector $Z$ in (\ref{iH1X}) can be regarded as the Weyl quantization of the $\mR^n$-valued function $x\mapsto  2\Theta \d_x H_1(x)$, where $\d_x(\cdot)$ denotes the gradient operator. Indeed, from (\ref{h}) it follows that $i h(\lambda)\lambda$ is the Fourier transform of the function $\d_xH_1$.

\section{WEIGHTED MEAN SQUARE FUNCTIONALS}\label{sec:quadro}
%%%%%%%%%%%%%%%%%%%%%%%%%%%%%%%%%%%%%%%%%%%%%%%%%%%%%%%%%%%%%%%%%%%%%%%%%%%%%%%%%%%%%%%%%%%%%%%%%%%

For what follows, consider a weighted mean square of the system variables:
\begin{equation}
\label{quadro}
    V:=
    \bE (X^{\rT} \Pi X) = \sum_{j,k=1}^{n} \pi_{jk} \bE(X_jX_k),
\end{equation}
where $\Pi:= (\pi_{jk})_{1\< j,k\< n}\in \mS_n$ is a given matrix. Here, the quantum expectation  $\bE \xi:= \Tr(\rho \xi)$ of a quantum variable $\xi$ is taken over a density operator $\rho$ with a product structure $\rho = \varpi \ox \upsilon$, where $\varpi$ is the initial quantum state of the system, and $\upsilon$ is the vacuum state of the external boson fields. We will now consider the time evolution of the quantity $V$.
%%%%%%%%%%%%%%%%%%%%%%%%%%%%%%%%%%%%%%%%%%%%%%%%%%%%%%%%%%%%%%%%%%%%%%%%%%%%%%%%%%%%%%%%%%%%%%%%%%%
\begin{lem}
\label{lem:Vdot}
For the perturbed quantum system governed by (\ref{dx1}), the quantity $V$ in (\ref{quadro}) satisfies a differential equation
\begin{align}
\nonumber
  \dot{V} =&    \Bra A^{\rT}\Pi + \Pi A , P\Ket + \Bra \Pi, BB^{\rT}\Ket\\
\label{Vdot}
   & + \Bra \Pi, \Re \bE(XZ^{\rT} + ZX^{\rT}) \Ket,
\end{align}
where
\begin{equation}
\label{P}
    P := \Re \bE(XX^{\rT})
\end{equation}
is the real part of the matrix of second moments of the system variables.
\end{lem}
%%%%%%%%%%%%%%%%%%%%%%%%%%%%%%%%%%%%%%%%%%%%%%%%%%%%%%%%%%%%%%%%%%%%%%%%%%%%%%%%%%%%%%%%%%%%%%%%%%%
\begin{proof}
Consider a matrix $S\in \mH_n$ of second moments of the system variables given by
\begin{equation}
\label{S}
  S := \bE(XX^{\rT}) = P + i\Theta,
\end{equation}
where use is made of (\ref{xCCR}) and (\ref{P}). Since $\Pi$ in (\ref{quadro}) is a real symmetric matrix, then
\begin{equation}
\label{VP}
    V = \bra \Pi, S\ket = \bra \Pi, P\ket.
\end{equation}
By combining the quantum It\^{o} lemma and the It{\^ o} product rules with (\ref{dx1}), it follows that
\begin{align}
\nonumber
    \rd (XX^{\rT}) =& (\rd X) X^{\rT} + X\rd X^{\rT} + (\rd X) \rd X^{\rT}\\
\nonumber
    =& ((A X + Z)\rd t+ B \rd W) X^{\rT}\\
\nonumber
    &+ X((X^{\rT} A^{\rT} + Z^{\rT})\rd t+ \rd W^{\rT} B^{\rT})\\
\nonumber
    & + B\rd W \rd W^{\rT} B^{\rT}\\
\nonumber
    = &
        (AXX^{\rT} + XX^{\rT}A^{\rT} + B\Omega B^{\rT}+XZ^{\rT} + ZX^{\rT})\rd t\\
%\nonumber
%    & + (XZ^{\rT} + ZX^{\rT})\rd t\\
\label{dXX}
    & +
    B\rd W X^{\rT} + X \rd W^{\rT} B^{\rT},
\end{align}
where use is also made of (\ref{WW}).
Since the external boson fields are assumed to be in the vacuum state, the averaging of the QSDE (\ref{dXX}) leads to a differential equation for the matrix $S$ in (\ref{S}):
\begin{equation}
\label{Sdot}
    \dot{S}
    =
        AS + SA^{\rT} + B\Omega B^{\rT}
    + \bE(XZ^{\rT} + ZX^{\rT}).
\end{equation}
In view of (\ref{WW}), (\ref{S}) and (\ref{VP}), the substitution of (\ref{Sdot}) into $\dot{V} = \bra \Pi, \Re \dot{S}\ket $ and using the duality relation $\bra \Pi, AP + PA^{\rT}\ket = \bra A^{\rT}\Pi + \Pi A, P\ket $
leads to (\ref{Vdot}), thus completing the proof of the lemma.
\end{proof}
%%%%%%%%%%%%%%%%%%%%%%%%%%%%%%%%%%%%%%%%%%%%%%%%%%%%%%%%%%%%%%%%%%%%%%%%%%%%%%%%%%%%%%%%%%%%%%%%%%%

Note that the first line in (\ref{Vdot}) corresponds to linear dynamics of the nominal open quantum harmonic oscillator, while the second line comes from the perturbation of the system Hamiltonian.

%%%%%%%%%%%%%%%%%%%%%%%%%%%%%%%%%%%%%%%%%%%%%%%%%%%%%%%%%%%%%%%%%%%%%%%%%%%%%%%%%%%%%%%%%%%%%%%%%%%
\section{DISSIPATION INEQUALITY}\label{sec:diss}
%%%%%%%%%%%%%%%%%%%%%%%%%%%%%%%%%%%%%%%%%%%%%%%%%%%%%%%%%%%%%%%%%%%%%%%%%%%%%%%%%%%%%%%%%%%%%%%%%%%

Following \cite[Section VIII]{VP_2012a}, we say that a matrix $L := (L_{jk})_{1\< j,k\< r}$ of linear operators on the underlying Hilbert space,  satisfying $L^{\dagger} = L$, is \emph{superpositive} if the self-adjoint operator $u^* L u := \sum_{j,k=1}^r \overline{u_j}u_k L_{jk}$ is positive semi-definite for any vector $u:= (u_j)_{1\< j\< r} \in \mC^r$. This notion extends the standard positive semi-definiteness from a single operator to a matrix of operators and will be written using the same symbol as $L \succcurlyeq 0$. Note that the superpositiveness $L \succcurlyeq  0$ implies that $\bE L\succcurlyeq 0$ for any density operator over which this expectation is taken. Indeed, $L \succcurlyeq 0$ implies that $u^* \bE L u = \bE(u^* L u)\> 0$ for any $u \in \mC^r$. Similarly to the usual positive semi-definiteness for single operators, the superpositiveness induces a partial  ordering for matrices of operators. An example of a superpositive matrix is provided by $XX^{\rT}$ because $u^* XX^{\rT}u = (u^* X)(u^* X)^{\dagger}\succcurlyeq 0$ for any $u \in \mC^n$.
Now, consider an operator inequality
\begin{equation}
\label{ZZ}
    ZZ^{\rT} \preccurlyeq \mu_1 \sum_{k=1}^d \Gamma_k XX^{\rT} \Gamma_k^\rT + \mu_0 I_n
\end{equation}
for the vector $Z$ from (\ref{dx1})
in the sense of superpositiveness. Here, $\Gamma_1, \ldots, \Gamma_d \in \mR^{n\x n}$ are fixed matrices, $\mu_1$ and $\mu_0$ are real constants, and $d$ is a positive integer. The following lemma derives a useful upper bound from (\ref{ZZ}).

%%%%%%%%%%%%%%%%%%%%%%%%%%%%%%%%%%%%%%%%%%%%%%%%%%%%%%%%%%%%%%%%%%%%%%%%%%%%%%%%%%%%%%%%%%%%%%%%%%%
\begin{lem}
\label{lem:XZ}
Suppose the vector $Z$ satisfies (\ref{ZZ}), with $\mu_1 >0$. Then
\begin{equation}
\label{XZ}
    XZ^{\rT} + ZX^{\rT} \preccurlyeq \sum_{k=0}^d \Gamma_k XX^{\rT} \Gamma_k^\rT +\frac{\mu_0}{\mu_1}I_n,
    \qquad
    \Gamma_0 :=  \sqrt{\mu_1}I_n.
\end{equation}
\end{lem}
%%%%%%%%%%%%%%%%%%%%%%%%%%%%%%%%%%%%%%%%%%%%%%%%%%%%%%%%%%%%%%%%%%%%%%%%%%%%%%%%%%%%%%%%%%%%%%%%%%%
\begin{proof}
Consider an auxiliary vector $Y$ of self-adjoint operators defined by
$
    Y := \sqrt{\mu_1} X - \frac{1}{\sqrt{\mu_1}} Z
$.
Since the matrix of operators
$
    YY^{\rT} = \mu_1 XX^{\rT} + \frac{1}{\mu_1} ZZ^{\rT} - XZ^{\rT}-ZX^{\rT}
$ is superpositive,
then
\begin{align*}
    XZ^{\rT}+ZX^{\rT} & \preccurlyeq \mu_1 XX^{\rT} + \frac{1}{\mu_1} ZZ^{\rT} \\
    & \preccurlyeq \mu_1 XX^{\rT} + \sum_{k=1}^{d} \Gamma_k XX^{\rT} \Gamma_k^\rT + \frac{\mu_0}{\mu_1}I_n.
\end{align*}
The last inequality follows from (\ref{ZZ}) and leads to (\ref{XZ}).
\end{proof}
%%%%%%%%%%%%%%%%%%%%%%%%%%%%%%%%%%%%%%%%%%%%%%%%%%%%%%%%%%%%%%%%%%%%%%%%%%%%%%%%%%%%%%%%%%%%%%%%%%%

We will now use Lemma~\ref{lem:XZ} in order to obtain a dissipation inequality for the mean square functional of the perturbed system.

%%%%%%%%%%%%%%%%%%%%%%%%%%%%%%%%%%%%%%%%%%%%%%%%%%%%%%%%%%%%%%%%%%%%%%%%%%%%%%%%%%%%%%%%%%%%%%%%%%%
\begin{lem}
Suppose the vector $Z$ in (\ref{dx1}) satisfies the operator inequality (\ref{ZZ}), with $\mu_1>0$. Also, let the weighting matrix $\Pi$ in (\ref{quadro}) be positive semi-definite and satisfy the following LMI
\begin{equation}
\label{Picond}
    A^{\rT} \Pi + \Pi A + \sum_{k=0}^d \Gamma_k^\rT \Pi \Gamma_k +\gamma \Pi\preccurlyeq  0,
\end{equation}
where $\gamma$ is a real constant, and $\Gamma_0$ is the matrix given by (\ref{XZ}). Then the quantity $V$ in (\ref{quadro}) satisfies a dissipation inequality
\begin{equation}
\label{Vdotineq}
  \dot{V} \<  -\gamma V + \Bra \Pi, BB^{\rT}\Ket  + \frac{\mu_0}{\mu_1}\Tr \Pi.
\end{equation}
\end{lem}
%%%%%%%%%%%%%%%%%%%%%%%%%%%%%%%%%%%%%%%%%%%%%%%%%%%%%%%%%%%%%%%%%%%%%%%%%%%%%%%%%%%%%%%%%%%%%%%%%%%

\begin{proof}
From (\ref{Picond}) and the positive semi-definiteness of the  matrix $P$ in (\ref{S}), it follows that the first inner product in (\ref{Vdot}) admits an upper bound
\begin{equation}
\label{up1}
  \Bra A^{\rT}\Pi + \Pi A , P\Ket
  \<
  \Bra
    -\sum_{k=0}^d \Gamma_k^{\rT} \Pi \Gamma_k-\gamma \Pi, P\Ket.
\end{equation}
By applying Lemma~\ref{lem:XZ} and using the monotonicity of the quantum  expectation with respect to the superpositiveness, it follows from (\ref{XZ}) that
\begin{align}
\nonumber
    \bE(XZ^{\rT} + ZX^{\rT})
    & \preccurlyeq
    \sum_{k=0}^d
    \bE(\Gamma_k XX^{\rT} \Gamma_k^\rT) + \frac{\mu_0}{\mu_1}I_n,\\
\label{EXZ}
    & \preccurlyeq \sum_{k=0}^d \Gamma_k S \Gamma_k^\rT + \frac{\mu_0}{\mu_1}I_n,
\end{align}
Here, use is made of (\ref{S}) and the fact that $\Gamma_0, \ldots, \Gamma_d$ are constant matrices. Therefore, since $\Pi\succcurlyeq 0$, then (\ref{EXZ}) leads to the following upper bound for the last inner product in (\ref{Vdot}):
\begin{align}
\nonumber
    \Bra \Pi, \Re \bE(XZ^{\rT} + ZX^{\rT}) \Ket
    &\< \Bra \Pi, \sum_{k=0}^d \Gamma_k P \Gamma_k^\rT + \frac{\mu_0}{\mu_1}I_n \Ket \\
\label{up2}
    & \<  \Bra \sum_{k=0}^d \Gamma_k^\rT \Pi \Gamma_k,  P \Ket + \frac{\mu_0}{\mu_1}\Tr \Pi.
\end{align}
It now remains to note that substitution of (\ref{up1}) and (\ref{up2}) into (\ref{Vdot}) establishes (\ref{Vdotineq}).
\end{proof}
%%%%%%%%%%%%%%%%%%%%%%%%%%%%%%%%%%%%%%%%%%%%%%%%%%%%%%%%%%%%%%%%%%%%%%%%%%%%%%%%%%%%%%%%%%%%%%%%%%%

Note that, in addition to the usual Lyapunov part $A^{\rT} \Pi + \Pi A$ (which comes from the nominal linear system), the LMI (\ref{Picond}) involves a linear operator
\begin{equation}
\label{Syl}
    \Pi \mapsto \sum_{k=0}^d \Gamma_k^{\rT} \Pi \Gamma_k
\end{equation}
on the space $\mS_n$, which, in view of (\ref{ZZ}), is associated with the perturbation of the Hamiltonian. The structure of the linear operator in (\ref{Syl}) is analogous to that of the generalized Sylvester equations \cite{GLAM_1992}.  Such operators are present in moment stability conditions for quasilinear quantum stochastic systems \cite[Section IX]{VP_2012b}. Furthermore, this operator structure resembles the Kraus form of quantum operations  \cite[pp. 360--373]{NC_2000}.

%%%%%%%%%%%%%%%%%%%%%%%%%%%%%%%%%%%%%%%%%%%%%%%%%%%%%%%%%%%%%%%%%%%%%%%%%%%%%%%%%%%%%%%%%%%%%%%%%%%
\section{MEAN SQUARE STABILITY}\label{sec:stab}
%%%%%%%%%%%%%%%%%%%%%%%%%%%%%%%%%%%%%%%%%%%%%%%%%%%%%%%%%%%%%%%%%%%%%%%%%%%%%%%%%%%%%%%%%%%%%%%%%%%

The following lemma provides sufficient conditions of mean square stability \cite{IVM2012,Ian2012robust} of the perturbed quantum system (\ref{dx1}).
%%%%%%%%%%%%%%%%%%%%%%%%%%%%%%%%%%%%%%%%%%%%%%%%%%%%%%%%%%%%%%%%%%%%%%%%%%%%%%%%%%%%%%%%%%%%%%%%%%%
\begin{lem}
\label{lem:mss}
Suppose the vector $Z$ in (\ref{dx1}) satisfies the operator inequality (\ref{ZZ}) for some $\mu_1 > 0$, $\Gamma_1, \ldots, \Gamma_d \in \mR^{n\x n}$ and $\mu_0\in \mR$. Also, suppose there exist a weighting matrix $\Pi\succ 0$ and a constant $\gamma > 0$ which satisfy the LMI (\ref{Picond}). Then the quantum stochastic system (\ref{dx1}) is mean square stable, with the upper limit of the quantity $V$ in (\ref{quadro}) satisfying
\begin{equation}
\label{Vlim}
    \limsup_{t\to +\infty} V
    \<
    \frac{1}{\gamma}\Big(\Bra \Pi, BB^{\rT}\Ket  + \frac{\mu_0}{\mu_1}\Tr \Pi\Big).
\end{equation}
\end{lem}
\begin{proof}
The upper bound (\ref{Vlim}) follows from the Gronwall-Bellman inequality applied to (\ref{Vdotineq}):
$$
    V(t)\< V(0)\re^{-\gamma t} + \Big(\Bra \Pi, BB^{\rT}\Ket  + \frac{\mu_0}{\mu_1}\Tr \Pi\Big) \frac{1-\re^{-\gamma t}}{\gamma},
$$
which holds for all times $t\> 0$, where use is made of the integral $\int_0^t \re^{-\gamma s}\rd s = \frac{1-\re^{-\gamma t}}{\gamma}$. It now remains to note that (\ref{Vlim}) implies mean square stability of the quantum system being considered in view of the assumption that $\Pi\succ 0$, whereby
$
    \bE (X^{\rT}X)\< \frac{V}{\lambda_{\min}(\Pi)}
$,
with $\lambda_{\min}(\Pi)>0$ denoting the smallest eigenvalue of the weighting matrix.  \end{proof}
%%%%%%%%%%%%%%%%%%%%%%%%%%%%%%%%%%%%%%%%%%%%%%%%%%%%%%%%%%%%%%%%%%%%%%%%%%%%%%%%%%%%%%%%%%%%%%%%%%%

In combination with the assumption of stability for the nominal linear quantum system,
Lemma~\ref{lem:mss} leads to the following sufficient conditions on the perturbation Hamiltonians which guarantee mean square stability of the perturbed system.
%%%%%%%%%%%%%%%%%%%%%%%%%%%%%%%%%%%%%%%%%%%%%%%%%%%%%%%%%%%%%%%%%%%%%%%%%%%%%%%%%%%%%%%%%%%%%%%%%%%
\begin{thm}
	\label{thm:ps}
	Suppose the matrix $A$ in (\ref{A}) is Hurwitz. Then the following set is nonempty:
	\begin{align}
    \nonumber
		\cP :=
        \Big\{ &
        (\mu_1, \Gamma_1,...,\Gamma_d) \in \mR \x (\mR^{n\x n})^d:\ d>0,\ \mu_1 >0, \\
\label{cP}
            & (\ref{Picond})\ {\rm  is\ satisfied\ for\ some}\ \Pi \succ 0,\ \gamma>0
            \Big\}.
	\end{align}
Furthermore, the perturbed quantum system (\ref{dx1}) is mean square stable for any perturbation Hamiltonian $H_1$ from the uncertainty set
	\begin{align}
    \nonumber
		\cU := \Big\{
    				& H_1\ {\rm given\ by}\  (\ref{H1}):\ Z\ {\rm satisfies}\ (\ref{ZZ}) \\
    \label{cU}
			   	    & {\rm for\ some}\ (\mu_1,\Gamma_1,...,\Gamma_d) \in \cP, \ \mu_0 \in \mR
					   \Big\},					
	\end{align}
	where $Z$ is the vector of operators associated with $H_1$ by (\ref{dx1}).
\end{thm}
%%%%%%%%%%%%%%%%%%%%%%%%%%%%%%%%%%%%%%%%%%%%%%%%%%%%%%%%%%%%%%%%%%%%%%%%%%%%%%%%%%%%%%%%%%%%%%%%%%

\begin{proof}
Since the matrix $A$ is Hurwitz, then there exist $\Pi\succ 0$ and $\gamma_0>0$ such that
$
    A^{\rT} \Pi + \Pi A + \gamma_0 \Pi \preccurlyeq 0
$.
Hence, by continuity, there exist (for any given positive integer $d$) sufficiently small $\mu_1>0$ and $\Gamma_1, \ldots, \Gamma_d \in \mR^{n\x n}$ such that the LMI (\ref{Picond}) is satisfied for the same $\Pi$ and a smaller $\gamma>0$  (that is, $0<\gamma < \gamma_0$). This proves that the set $\cP$ in (\ref{cP}) is indeed nonempty. Now, the property, that the perturbed system (\ref{dx1}) is mean square stable for any perturbation Hamiltonian $H_1$ belonging to the class $\cU$ in (\ref{cU}), was established in Lemma~\ref{lem:mss}.
\end{proof}

%%%%%%%%%%%%%%%%%%%%%%%%%%%%%%%%%%%%%%%%%%%%%%%%%%%%%%%%%%%%%%%%%%%%%%%%%%%%%%%%%%%%%%%%%%%%%%%%%%%
\section{TECHNIQUES FOR VERIFYING THE OPERATOR INEQUALITY}\label{sec:exinequ}
%%%%%%%%%%%%%%%%%%%%%%%%%%%%%%%%%%%%%%%%%%%%%%%%%%%%%%%%%%%%%%%%%%%%%%%%%%%%%%%%%%%%%%%%%%%%%%%%%%%

We will now discuss several techniques for verifying the operator inequality (\ref{ZZ}) which plays a central role for the uncertainty class $\cU$
in (\ref{cU}).
%
%
%Sufficient conditions for the guaranteed mean square stability of the perturbed system is studied in view of the bounds for perturbation Hamiltonian in this section.
%
%The dissipation inequality (\ref{Vdotineq}) in Section~\ref{sec:diss} which provides us with sufficient conditions of mean square stability in Section~\ref{sec:stab} are developed based on the assumption over the feasibility of the inequality (\ref{ZZ}). In view of the results for the sufficient conditions of mean square stability, which is given in Theorem~\ref{thm:ps}, the inequality (\ref{ZZ}) exerts a bound for the admissible Hamiltonian perturbations which naturally guarantees the preservation of the mean square stability. The conditions that the inequality (\ref{ZZ}) put on the Hamiltonian perturbation $H_1(x)$ or equivalently on $h(\lambda)$ in (\ref{h}) is presented in what follows.
%
%For an additional insight into the structure of the uncertainty class $\cU$ in (\ref{cU}),
The following lemma reformulates this  operator inequality in terms of the Weyl quantization of the perturbation Hamiltonian.

%In this way, Lemma~\ref{lem:p_ineq} provides a corresponding between the inequality (\ref{ZZ}) and the members of the set of perturbation Hamiltonians $\mathcal{U}$, introduced in (\ref{cU}), that preserves mean square stability.
%%%%%%%%%%%%%%%%%%%%%%%%%%%%%%%%%%%%%%%%%%%%%%%%%%%%%%%%%%%%%%%%%%%%%%%%%%%%%%%%%%%%%%%%%%%%%%%%%%
\begin{lem}
\label{lem:p_ineq}
The operator inequality (\ref{ZZ}) for the vector $Z$ in (\ref{dx1}) is representable in the form
\begin{equation}
\label{ZZ_Weyl}
    4 \Theta
	\int_{\mR^n}
        \wt{h}(\lambda)
            \re^{i\lambda^{\rT}X}
    \rd \lambda
    \Theta
    \preccurlyeq \mu_1 \sum_{k=1}^d \Gamma_k XX^{\rT} \Gamma_k^\rT + \mu_0 I_n.
\end{equation}
Here, the function $\wt{h}: \mR^n \to \mC^{n\x n}$ is expressed in terms of the Fourier transform $h$ from (\ref{h}) as
\begin{equation}
\label{htilde}
\wt{h}(\lambda):=
      \int_{\mR^n}
        h(\lambda-\tau)h(\tau)
        (\lambda-\tau)
        \tau^\rT
	        \re^{i\tau^\rT \Theta \lambda}
    \rd \tau.
\end{equation}
\end{lem}
%%%%%%%%%%%%%%%%%%%%%%%%%%%%%%%%%%%%%%%%%%%%%%%%%%%%%%%%%%%%%%%%%%%%%%%%%%%%%%%%%%%%%%%%%%%%%%%%%%%%%
\begin{proof}
From (\ref{iH1X}) and the antisymmetry of the CCR matrix $\Theta$, it follows that
\begin{align}
\nonumber
	ZZ^\rT&=
	4
    \Theta
	\int_{\mR^n}
        h(\lambda)
        \lambda
            \re^{i\lambda^{\rT}X}
    \rd \lambda
    \int_{\mR^n}
	h(\tau)
        \tau^\rT
            \re^{i\tau^{\rT}X}
        \rd \tau\,
                \Theta
    \\
\nonumber
    &=
    4\Theta
	\int_{\mR^{2n}}
        h(\lambda)h(\tau)
        \lambda
        \tau^\rT
            \re^{i\lambda^{\rT}X}
            \re^{i\tau^{\rT}X}
    \rd \lambda
    \rd \tau\,\Theta    \\
\nonumber
    &=
    4\Theta
	\int_{\mR^{2n}}
        h(\lambda)h(\tau)
        \lambda
        \tau^\rT
	        \re^{i\tau^\rT \Theta \lambda}
            \re^{i(\lambda+\tau)^{\rT}X}
    \rd \lambda
    \rd \tau\,
            \Theta\\
\label{ZZW}
    &= 	
    4\Theta \int_{\mR^n}
        \wt{h}(\lambda)
            \re^{i\lambda^{\rT}X}
    \rd \lambda\,
    \Theta,
\end{align}
where the function $\wt{h}$ is given by (\ref{htilde}).
Here, use is made of the Baker-Campbell-Hausdorff formula \cite[pp. 40]{M_1998} and the CCRs (\ref{xCCR}), whereby
$$
		    \re^{i\lambda^{\rT}X}
            \re^{i\tau^{\rT}X}
            =
            \re^{\frac{1}{2}[i\lambda^{\rT}X,i\tau^{\rT}X]}
            \re^{i(\lambda+\tau)^{\rT}X}
            =
            \re^{i\tau^\rT \Theta \lambda}
            \re^{i(\lambda+\tau)^{\rT}X},
$$
and the standard change of variables $(\lambda,\tau)\mapsto (\lambda-\tau,\tau)$ for convolution integrals. Substitution of (\ref{ZZW}) into (\ref{ZZ}) leads to (\ref{ZZ_Weyl}).
\end{proof}
%%%%%%%%%%%%%%%%%%%%%%%%%%%%%%%%%%%%%%%%%%%%%%%%%%%%%%%%%%%%%%%%%%%%%%%%%%%%%%%%%%%%%%%%%%%%%%%%%%

%The remaining parts of this section is devoted to derive explicit inequalities for a class of Hamiltonian perturbations which belongs to the set $\cU$ in (\ref{cU}). To achieve such an inequality use is made of

The following lemma, which is given here for completeness,  extends the ordering of real-valued functions of a real variable to the case when they are evaluated at a self-adjoint operator.

%%%%%%%%%%%%%%%%%%%%%%%%%%%%%%%%%%%%%%%%%%%%%%%%%%%%%%%%%%%%%%%%%%%%%%%%%%%%%%%%%%%%%%%%%%%%%%%%%%
\begin{lem}
    \label{lem:ext_ieq_f}
    Suppose $f,g: [a,b] \to \mR$ are continuous functions satisfying $f(z) \< g(z)$ in an interval $a\< z\< b$. Then $f(K) \preccurlyeq g(K)$ for any  self-adjoint operator $K$ whose spectrum is contained by this interval.

\end{lem}
%%%%%%%%%%%%%%%%%%%%%%%%%%%%%%%%%%%%%%%%%%%%%%%%%%%%%%%%%%%%%%%%%%%%%%%%%%%%%%%%%%%%%%%%%%%%%%%%%%

    \begin{proof}
        The spectral theorem (see, for example,  \cite[pp. 263]{reed1980}) implies that a self-adjoint operator $K$, described in the lemma, is representable as
        $
            K = \int_a^b z\nu(\rd z)
        $,
        where $\nu$ is a projection-valued measure on the interval $[a,b] \subset \mR$.  Therefore, since
        $f(z)-g(z) \< 0$ for all $z \in [a,b]$, then the operator
        $
           g(K)-f(K) = \int_a^b (g(z)-f(z)) \nu(\rd z)
        $
        is positive semi-definite, and hence, $f(K) \preccurlyeq g(K)$.
    \end{proof}
%%%%%%%%%%%%%%%%%%%%%%%%%%%%%%%%%%%%%%%%%%%%%%%%%%%%%%%%%%%%%%%%%%%%%%%%%%%%%%%%%%%%%%%%%%%%%%%%%%

Note that Lemma~\ref{lem:ext_ieq_f} is useful for investigating the superpositiveness of a matrix of operators which are collinear to a real-valued function of a given self-adjoint operator. The following lemma studies a combined effect of several perturbations of the Hamiltonian on the operator inequality (\ref{ZZ}).
\begin{lem}
	\label{lem:decomp}
	Suppose the perturbation vector $Z$ in (\ref{dx1}) is decomposed as
	\begin{equation}
		\label{equ:decomp}
		Z = \sum_{k=1}^{d} c_k Z_{k},
	\end{equation}
    where $c_k$ are real coefficients and $Z_1, \ldots, Z_d$ are $n$-dimensional vectors of self-adjoint operators on the underlying Hilbert space. Also, suppose the operator inequalities
	\begin{equation}
		\label{ineq:decomp}
			Z_k Z_k^{\rT} \preccurlyeq \mu_{1} \Phi_k X X^\rT \Phi_k^\rT + \mu_{0k} I_n,
        \qquad
        k=1,\ldots,d,
	\end{equation}
are satisfied
    for some matrices $\Phi_k\in \mR^{n\x n}$ and real constants $\mu_{1} > 0$ and $\mu_{0k}$.
	Then the vector $Z$ in (\ref{equ:decomp}) satisfies (\ref{ZZ}) with the same constant $\mu_1$ and the following parameters $\Gamma_1, \ldots, \Gamma_d $ and $\mu_0$:
	\begin{equation}
		\label{equ:dec:Gamma_k_equ:dec:mu0}
    			 \Gamma_k  := \sqrt{\sigma_k}c_k\Phi_k,
                \qquad
		\mu_0 := \sum_{k=1}^{d} \sigma_k c_k^2 \mu_{0k},
	\end{equation}
where the coefficients $\sigma_k$ are computed in terms of arbitrary positive scalars $\nu_{jk}=\nu_{kj}$ as
    \begin{equation}
    \label{sig}
        \sigma_k := 1 + \sum_{j=1}^{k-1} \nu_{jk} + \sum_{j=k+1}^d \frac{1}{\nu_{jk}},
        \qquad
        k = 1,\ldots, d.
    \end{equation}
\end{lem}

%%%%%%%%%%%%%%%%%%%%%%%%%%%%%%%%%%%%%%%%%%%%%%%%%%%%%%%%%%%%%%%%%%%%%%%%%%%%%%%%%%%%%%%%%%%%%%%%%%%%%%%%
\begin{proof}
	In view of (\ref{equ:decomp}),
	\begin{align}
        %\nonumber
		Z Z^\rT %&= \sum_{j,k=1}^{d} c_jc_k Z_j Z_k^\rT \\
\label{ZZZ}
				&= \sum_{k=1}^{d} c_k^2 Z_k Z_k^\rT + \sum_{1 \< j < k \< d} c_jc_k(Z_j Z_k^\rT + Z_k Z_j^\rT ).
	\end{align}
By applying the completion-of-the-square technique, used in Lemma~\ref{lem:XZ}, to the vectors of operators $c_jZ_j$ and $c_kZ_k$, it follows that
	\begin{equation}		
    \label{ZZZZ}
		c_jc_k(Z_j Z_k^\rT + Z_k Z_j^\rT) \preccurlyeq \frac{c_j^2}{\nu_{jk}} Z_j Z_j^\rT + \nu_{jk} c_k^2  Z_k Z_k^\rT,
	\end{equation}
where $\nu_{kj}=\nu_{jk}$ are arbitrary positive scalars.  Substitution of (\ref{ZZZZ}) into (\ref{ZZZ}) and combining the result with (\ref{ineq:decomp}) leads to the following operator inequality (\ref{ZZ}) for the vector $Z$ in (\ref{equ:decomp}):
	\begin{align*}
		\nonumber
		Z Z^\rT &\preccurlyeq \sum_{k=1}^{d} c_k^2 Z_k Z_k^\rT + \sum_{1 \< j < k \< d} \Big(\frac{c_j^2}{\nu_{jk}} Z_j Z_j^\rT + \nu_{jk} c_k^2  Z_k Z_k^\rT\Big)\\
\label{ineq:decomp:intr}
    & =
    \sum_{k=1}^d \sigma_k c_k^2 Z_kZ_k^{\rT}
				\preccurlyeq \mu_1 \sum_{k=1}^d \Gamma_k X X^\rT \Gamma_k^\rT+ \mu_0 I_n,
	\end{align*}
	where use is made of the notation (\ref{equ:dec:Gamma_k_equ:dec:mu0}) and (\ref{sig}).
\end{proof}

%%%%%%%%%%%%%%%%%%%%%%%%%%%%%%%%%%%%%%%%%%%%%%%%%%%%%%%%%%%%%%%%%%%%%%%%%%%%%%%%%%%%%%%%%%%%%%%
\section{ILLUSTRATIVE EXAMPLES} \label{sec:exmpl}
%%%%%%%%%%%%%%%%%%%%%%%%%%%%%%%%%%%%%%%%%%%%%%%%%%%%%%%%%%%%%%%%%%%%%%%%%%%%%%%%%%%%%%%%%%%%%%%

%In the context of quantum physics, some of the most peculiar behavioural phenomena of quantum systems have been associated with systems which possess partly periodic Hamiltonian functions (see, \cite{?}).
The following examples aim to demonstrate an application of Theorem~\ref{thm:ps} to the robust mean square stability analysis of the open quantum system (\ref{dx1}) when the perturbation Hamiltonian is a trigonometric polynomial of the system variables.

\begin{exmp}
\label{exp:cosf}
Suppose $H_1:= \cos(\lambda_0^\rT X)$, where  $\lambda_0 \in \mR^n$ is a constant vector of spatial frequencies. This perturbation Hamiltonian results from the Weyl quantization (\ref{H1}) of the function $\cos(\lambda_0^\rT x)$ whose Fourier transform (\ref{h}) is given by $h(\lambda)= \frac{1}{2} (\delta(\lambda-\lambda_0) + \delta(\lambda+\lambda_0))$, with $\delta(\cdot)$ denoting the $n$-dimensional Dirac delta-function.  Substitution of $h$ into (\ref{iH1X}) yields $Z = -2\Theta \lambda_0 \sin(\lambda_0^{\rT}X)$, and hence,
\begin{equation}
\label{ZZsin}
    ZZ^{\rT} = -4 \Theta \lambda_0 \sin^2(\lambda_0^\rT X)\lambda_0^\rT \Theta,
\end{equation}
which can also be obtained by using Lemma~\ref{lem:p_ineq}. Here, use is also made of the antisymmetry of the CCR matrix $\Theta$.
Now, application of Lemma~\ref{lem:ext_ieq_f} leads to $\sin^2(\lambda_0^{\rT}X)\preccurlyeq (\lambda_0^{\rT}X)^2 = \lambda_0^{\rT} XX^{\rT}\lambda_0$, which, in combination with (\ref{ZZsin}), implies that
$$
	ZZ^\rT \preccurlyeq -4 \Theta \lambda_0 \lambda_0^{\rT} XX^{\rT}\lambda_0 \lambda_0^\rT \Theta
		   =  -4 \Theta \lambda_0 \lambda_0^\rT X X^\rT \lambda_0  \lambda_0^\rT \Theta.
$$
Therefore, the operator inequality (\ref{ZZ}) is satisfied for $d=1$ with $\Gamma_1 = \frac{2}{\sqrt{\mu_1}} \Theta \lambda_0 \lambda_0^\rT$ and $\mu_0=0$. According to Theorem~\ref{thm:ps}, the perturbation Hamiltonian $H_1$ being considered belongs to the uncertainty class $\cU$ with respect to which the system is robustly mean square stable, provided the following LMI
$$
    A^{\rT}\Pi + \Pi A + (\mu_1 + \gamma) \Pi + \frac{4}{\mu_1} \|\Theta \lambda_0\|_{\Pi}^2 \lambda_0  \lambda_0^{\rT}\preccurlyeq 0
$$
(obtained from (\ref{Picond})) holds for some $\Pi\succ 0$, $\mu_1>0$, $\gamma>0$.
\end{exmp}
%%%%%%%%%%%%%%%%%%%%%%%%%%%%%%%%%%%%%%%%%%%%%%%%%%%%%%%%%%%%%%%%%%%%%%%%%%%%%%%%%%%%%%%%%%%%%%%%%%%%%%

\begin{exmp}
\label{exp:trigf}
Consider a perturbation Hamiltonian $H_1:= \sum_{k=1}^{d} r_k \cos(\lambda_k^\rT X+\phi_k)$, where $\lambda_k\in \mR^n$  are given vectors of spatial frequencies, $r_k> 0$ are amplitudes and $0\< \phi_k< 2\pi$ are initial phases which form complex amplitudes $a_k:= r_k \re^{i\phi_k}$. The corresponding Fourier transform in (\ref{h}) is   $h(\lambda)= \frac{1}{2} \sum_{k=1}^{d} (a_k \delta(\lambda-\lambda_k) + \overline{a_k} \delta(\lambda+\lambda_k))$. Then the perturbation vector $Z$ admits the decomposition (\ref{equ:decomp}) with unit coefficients $c_k = 1$, where, in view of Example~\ref{exp:cosf},
\begin{align*}
    Z_k &= i r_k  \Theta \lambda_k \big( \re^{i(\lambda_k^\rT X+\phi_k)}-\re^{-i(\lambda_k^\rT X+\phi_k)}\big)\\
        &= -2 r_k  \Theta \lambda_k \sin(\lambda_k^\rT X+\phi_k),
        \qquad
        k = 1,\ldots, d.
\end{align*}
Hence, by applying Lemma~\ref{lem:ext_ieq_f} twice, it follows that
\begin{align*}
    Z_k Z_k^\rT
                &= -4 r_k^2 \Theta \lambda_k \sin^2(\lambda_k^\rT X+\phi_k) \lambda_k^\rT \Theta \\
                &\preccurlyeq -4 r_k^2 \Theta \lambda_k (\lambda_k^\rT X+\phi_k)^2\lambda_k^\rT \Theta  \\
                &\preccurlyeq  -4 r_k^2 \Theta \lambda_k \Big((1+\omega_k) \lambda_k^{\rT}XX^{\rT}\lambda_k + \Big(1+\frac{1}{\omega_k}\Big)\phi_k^2\Big)\lambda_k^\rT \Theta,
\end{align*}
where $\omega_k$ are arbitrary positive real parameters.  Therefore, the inequalities (\ref{ineq:decomp}) are satisfied with the following parameters:
\begin{equation}
\label{Phimu}
	\Phi_{k}  \!:=\! 2r_k \sqrt{\frac{1+\omega_k}{\mu_1}}\, \Theta \lambda_k \lambda_k^\rT,\ \ \
	\mu_{0k}      \!:=\! 4 r_k^2\phi_k^2\frac{1+\omega_k}{\omega_k} |\Theta \lambda_k|^2.\!\!\!
\end{equation}
These can be employed in order to find the parameters $\Gamma_k$ and $\mu_0$ according to (\ref{equ:dec:Gamma_k_equ:dec:mu0}) and (\ref{sig}) of Lemma~\ref{lem:decomp}, and then proceed to the robust mean square stability analysis through the LMI (\ref{Picond}) as described in Theorem~\ref{thm:ps}.
%
%These parameters can be used in order to find
%As the result,
%where, Subsequently,  in turn can be derived by (\ref{equ:dec:Gamma_k_equ:dec:mu0}), by taking arbitrary weightings $\nu_{ij}=\nu_{ji}>0$ based on the result of Lemma~\ref{lem:decomp}. Then, according to Theorem~\ref{thm:ps}, for a given stable nominal quantum system (\ref{dx1}), $H_1(X) \in \mathcal{U}$ if $(\Gamma_1,...,\Gamma_d,\mu_1) \in \mathcal{P}$.
\end{exmp}
%%%%%%%%%%%%%%%%%%%%%%%%%%%%%%%%%%%%%%%%%%%%%%%%%%%%%%%%%%%%%%%%%%%%%%%%%%%%%%%%%%%%%%%%%%%%%%%%%%%%%%%%
\begin{exmp}
    Let  $E := H_1-\sum_{k=1}^d r_k \cos(\lambda_k^\rT X+\phi_k)$ be an error of approximation of the perturbation Hamiltonian by a trigonometric polynomial from Example~\ref{exp:trigf}. Suppose its contribution $i[E,X]$ to the perturbation vector $Z$ satisfies
    $$
       -[E,X][E,X]^{\rT} \preccurlyeq \mu_1 \Gamma XX^\rT \Gamma^\rT + \mu I_n,
    $$
    with $\Gamma \in \mR^{n\x n}$ and $\mu \in \mR$. A combination of Lemma~\ref{lem:decomp} with the results of Example~\ref{exp:trigf} leads to an augmented set of parameters which consists of  $\Phi_1, \ldots, \Phi_d$, $\mu_{01}, \ldots, \mu_{0d}$ from (\ref{Phimu}) and $\Phi_{d+1}:= \Gamma$, $\mu_{0,d+1}:= \mu$. The remaining part of the robust stability analysis procedure is carried out as before.
%   \begin{align*}
%	  \Gamma_{k} &= 2 r_k \sqrt{(1+\omega_k)}\Theta \lambda_k \lambda_k^\rT,\\
%      \mu_{0k}   &= 4 r_k^2\phi_k^2(1+\frac{1}{\omega_k}) \bar{\sigma}^2 (\Theta \lambda_k \lambda_k^\rT),\\
%      \Gamma_{n_0} &= \Gamma_{e_{n_0}}, \\
%      \mu_{0n_0} &= \mu_{e_0n_0},
%\end{align*}
%    for any $\omega_k >0$ and $c_k=1$ for $k=1,...,d$.
%$\Gamma_k$ and $\mu_0$ in turn can be calculated by (\ref{equ:dec:Gamma_k_equ:dec:mu0}), by taking  real constants $\nu_{ij}=\nu_{ji}>0$ based on the result of Lemma~\ref{lem:decomp}. According to Theorem~\ref{thm:ps}, for a given stable nominal quantum system (\ref{dx1}), $H_1(X) \in \mathcal{U}$ if $(\Gamma_1,...,\Gamma_d,\mu_1) \in \mathcal{P}$.
\end{exmp}

%%%%%%%%%%%%%%%%%%%%%%%%%%%%%%%%%%%%%%%%%%%%%%%%%%%%%%%%%%%%%%%%%%%%%%%%%%%%%%%%%%%%%%%%%%%%%%%%%%%%%%%%
\section{CONCLUSION}\label{sec:conclusion}
%%%%%%%%%%%%%%%%%%%%%%%%%%%%%%%%%%%%%%%%%%%%%%%%%%%%%%%%%%%%%%%%%%%%%%%%%%%%%%%%%%%%%%%%%%%%%%%%%%%%%%%%

In this paper, we have presented a novel model for perturbations of Hamiltonians in a Weyl quantization form for a class of open quantum stochastic systems with linear coupling to the external boson fields. The time evolution of weighted mean square functionals of the system variables  and a related dissipation inequality have been studied in order to develop sufficient conditions for robust  mean square stability. An admissible class of Hamiltonian perturbations of a given stable linear quantum system  has been formulated to guarantee stability of the resulting perturbed system. We have also discussed feasibility of these conditions in terms of the Weyl quantization model. This approach to the modelling and robust stability analysis of uncertain quantum stochastic systems has been demonstrated for several examples with Hamiltonian perturbations in the form of trigonometric polynomials of system variables.

\end{document}